\documentclass{eptcs}
\usepackage{amsmath, amsthm, amssymb, xspace}
\usepackage[all]{xy}

\providecommand{\event}{QPL 2012}

\newcommand{\cat}{\ensuremath{\mathbf{C}}\xspace} 
\newcommand{\Coalg}[1]{\ensuremath{\Cat{CoAlg}(#1)}}
\newcommand{\Cat}[1]{\ensuremath{\mathbf{#1}}\xspace} 
\newcommand{\Sets}{\Cat{Sets}}
\newcommand{\Conv}{\Cat{Conv}}
\newcommand{\Hilbisometry}{\Cat{Hilb_{Isomet}}}
\newcommand{\D}{\mathcal{D}}
\newcommand{\DM}{\mathcal{D}\!\mathcal{M}}
\newcommand{\Ef}{\mathcal{E}\!f}
\newcommand{\beh}[1]{\mathrm{beh}_{#1}} 
\newcommand{\tail}{\mathrm{tail}}

\newcommand{\id}{\mathrm{id}} 
\newcommand{\effect}{\epsilon}
\newcommand{\unops}{S}
\newcommand{\meass}{E}

\newcommand{\pow}[1]{\mathcal{P}(#1)} 
\newcommand{\ost}{\tfrac{1}{\sqrt{2}}}

\newcommand{\bra}[1]{\langle #1 |}
\newcommand{\ket}[1]{| #1 \rangle}
\newcommand{\bk}[3]{\langle #1 | #2 | #3 \rangle}

\newcommand{\ua}{\mkern-5mu \uparrow}
\newcommand{\da}{\mkern-5mu \downarrow}
\newcommand{\ku}[1]{\ket{\ua \! #1}}
\newcommand{\kd}[1]{\ket{\da \! #1}}
\newcommand{\bu}[1]{\bra{\uparrow \! #1}}
\newcommand{\bd}[1]{\bra{\downarrow \! #1}}

\DeclareMathOperator{\supp}{supp}
\DeclareMathOperator{\tr}{tr}

\newtheorem{thm}{Theorem}
\newtheorem{lemma}[thm]{Lemma}

\newtheorem{proposition}[thm]{Proposition}
\theoremstyle{definition}
\newtheorem{definition}[thm]{Definition}
\newtheorem{example}[thm]{Example}

\theoremstyle{remark}

\title{Coalgebraic Quantum Computation}
\author{Frank Roumen
\institute{Institute for Mathematics, Astrophysics, and Particle Physics
(IMAPP) \\ Radboud University Nijmegen}
\email{F.Roumen@math.ru.nl}
}

\def\titlerunning{Coalgebraic Quantum Computation}
\def\authorrunning{Frank Roumen}

\begin{document}

\maketitle

\begin{abstract}
   Coalgebras generalize various kinds of dynamical systems occuring in
   mathematics and computer science. Examples of systems that can be modeled
   as coalgebras include automata and Markov chains. We will present a
   coalgebraic representation of systems occuring in the field of quantum
   computation, using convex sets of density matrices as state spaces.
   This will allow us to derive a method to convert quantum
   mechanical systems into simpler probabilistic systems with the same
   probabilistic behaviour. 
\end{abstract}

\section{Introduction}\label{sec_intro}

For studying complex computational systems, it is often helpful to use an
abstract description of the systems. This helps to focus on the most important
parts of the system under consideration and see similarities and differences
between distinct kinds of systems. The notion of a coalgebra, which originated
in category theory, gives such an abstract view on dynamical systems. There
is a large class of systems that can be described using coalgebras, including
finite automata, Turing machines, Markov chains, and differential equations.
Moreover one can reason about these system using a unified theory. An overview
can be found in \cite{rutten_universalcoalg}.

Our aim is to describe systems occuring in the field of quantum computation in
the coalgebraic framework. This will facilitate comparison between quantum
systems and, for example, deterministic and probabilistic systems. It will
also enable us to apply facts from the general theory of coalgebras to quantum
mechanical systems. In particular we will see what the minimization procedure
from \cite{abhkms_minimization} amounts to for quantum coalgebras.

The use of categories in the foundations of quantum physics was initiated in
\cite{ac_protocols}, via tensor categories, and in \cite{bi_topos}, via topoi.
Representation of quantum systems with coalgebras is also considered in
\cite{abramsky_chucoalg}. However, there are several differences between
\cite{abramsky_chucoalg} and our work. We focus on the dynamics of systems via
unitary operators, whereas \cite{abramsky_chucoalg} models the dynamics of
iterated measurements.

The outline of this paper is as follows. In Section~\ref{sec_coalgebras} we
discuss preliminaries on coalgebras. To illustrate the theory of coalgebras,
Section~\ref{sec_probabilisticsystems} contains a coalgebraic description of
probabilistic systems. Some of the constructions in that section are also
necessary for understanding quantum systems coalgebraically. The main original
contributions are in Sections~\ref{sec_model} and \ref{sec_minimization}.
Section~\ref{sec_model} shows how to represent quantum systems as coalgebras,
and discusses the role of final coalgebras in this framework. Finally, in
Section~\ref{sec_minimization}, the minimization procedure from
Section~\ref{sec_coalgebras} is applied to quantum coalgebras.

\section{Coalgebras}\label{sec_coalgebras}

In this section we will present some preliminary material on coalgebras, which
are abstract generalizations of state-based systems. Let $F$ be an endofunctor
on a category \cat. An \emph{$F$-coalgebra} consists of an object $X\in \cat$
and a morphism $c: X \to F(X)$ in \cat.  The object $X$ is called the
\emph{state space} of the coalgebra, and the morphism $c$ is called the
\emph{dynamics}.  The functor $F$ in the definition is a parameter that
determines the structure of the dynamics, and hence the kind of system.
Coalgebras for a fixed endofunctor constitute a category, in which a morphism
from $c : X \to F(X)$ to $d : Y \to F(Y)$ is a morphism $f: X \to Y$ in \cat
making the following diagram commute.
\[ \xymatrix{ X \ar[r]^f \ar[d]_c & Y \ar[d]^d \\
              F(X) \ar[r]_{F(f)}  & F(Y)} \]
The category of $F$-coalgebras and homomorphisms is denoted $\Coalg{F}$.

An $F$-coalgebra $\omega : \Omega \to F(\Omega)$ is said to be \emph{final} if
it is a final object in the category $\Coalg{F}$, i.e. for every coalgebra $c:
X \to F(X)$ there exists a unique coalgebra morphism from $c$ to $\omega$.
This unique morphism is called the \emph{behaviour morphism} of $c$ and is
denoted $\beh{c} : X \to \Omega$.

\begin{example}
    \label{ex_automata}
    In theoretical computer science, finite automata provide a mechanical way
    to describe languages. An \emph{alphabet} is a finite set, whose elements
    we call \emph{letters} or \emph{symbols}. The set of finite sequences, or
    \emph{words}, with entries in an alphabet $A$ is written as $A^*$.  This
    set of words forms a monoid, where the monoid operation is concatenation
    and the empty word $\varepsilon$ acts as an identity element. The monoid
    $A^*$ is the free monoid over $A$. A \emph{language} over $A$ is a subset
    of $A^*$.
    A \emph{deterministic automaton} over an alphabet $A$ consists of a set
    $X$ of states, a transition function $\delta : X \times A \to X$, and a
    subset $U\subseteq X$ of accepting states. 
    The transition function $\delta : X \times A \to X$ can be extended to a
    function from $X \times A^*$ to $X$, also denoted $\delta$, by defining
    $\delta(x,\varepsilon) = x$ and $\delta(x,au) =
    \delta(\delta(x,a),u)$. 

    Automata are often graphically represented by their state diagrams. In a
    state diagram, each state of the automaton is drawn as a circle. A
    transition $\delta(x,a) = y$ is indicated by an arrow, labeled with the
    letter $a$, from the circle $x$ to the circle $y$. 
    Accepting states are drawn as double circles. As an example, consider the
    automaton over $A = \{a,b\}$ with the following diagram.
    \[
    \entrymodifiers={++[o][F-]} \SelectTips{cm}{} 
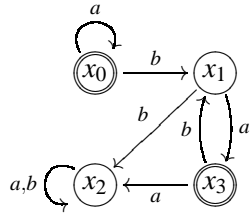
\xymatrix @-1pc { 
     *++[o][F=]{x_0} \ar@(ul,ur)[]^a \ar[rr]^b & *\txt{}& {x_1}
    \ar[ddll]_b \ar@/^/[dd]^a \\ *\txt{} \\  {x_2} \ar@(ul,dl)_{a,b}
    & *\txt{} & *++[o][F=]{x_3} \ar@/^/[uu]^b \ar[ll]^a }\]
    Here, the set of
    states is $X = \{x_0,x_1,x_2,x_3\}$, the transition function is completely
    determined by the arrows in the diagram, and the subset of accepting
    states is $\{x_0,x_3\}$.

    The transition function and the subset of accepting states of an automaton
    can be merged into one function of type $X \to 2 \times X^A$ Thus
    deterministic automata are coalgebras for the endofunctor $F(X) = 2
    \times X^A$ on $\Sets$. This functor $F$ has a final coalgebra whose
    underlying state space is the set $\pow{A^*}$ of all languages over $A$.
    To endow this set with an automaton structure, define the transition
    function by $\delta(L,a) = \{u \in A^* \mid au \in L \}$, and let the
    subset of accepting states be $\{L \in \pow{A^*} \mid \varepsilon \in
    L\}$. Let $c : X \to 2 \times X^A$ be an arbitrary automaton, with
    transition function $\delta$ and accepting states $U$. Then the
    behaviour morphism $\beh{c} : X \to \pow{A^*}$ assigns to a state $x$ the
    language $\{ u \in A^* \mid \delta(x,u) \in U \}$. Thus the coalgebraic
    description gives a convenient way to characterize the language recognized
    by an automaton.
\end{example}

There are several other systems that can be modeled as coalgebras in such a
way that the morphism into the final coalgebra corresponds to the observable
behaviour of the system, see \cite{rutten_universalcoalg} for more examples.

A morphism into the final coalgebra provides a method to obtain an external
description of a system, given an internal description. We will now turn our
attention to the reverse problem: if we know the behaviour of a system, how do
we find a coalgebra having that behaviour? Of course, in practice there are
many coalgebras with the same behaviour. We are often interested in the most
efficient one, i.e. the coalgebra with the smallest state space among those
with the same behaviour. Finding the coalgebra with a minimal state space is
known as the problem of minimization. Here we will focus on the special case
of minimizations of automata.

Minimal realizations of automata were first constructed in \cite{nerode}. This
was generalized to categorical settings in \cite{goguen}, see also
\cite{abhkms_minimization,ruttenautomatacoinduction} for the coalgebraic
version.

\begin{definition}
    A \emph{subcoalgebra} of a coalgebra $c : X \to F(X)$ in $\Sets$ is a
    coalgebra $d : Y \to F(Y)$ with $Y\subseteq X$ for which the inclusion map
    $Y \hookrightarrow X$ is a coalgebra morphism.
\end{definition}

\begin{definition}
    Let $c : X \to F(X)$ be a coalgebra in $\Sets$, and $S\subseteq X$. The
    subcoalgebra of $c$ \emph{generated} by $S$ is the smallest subcoalgebra
    $\langle S \rangle$ of $c$ whose state space includes $S$. If $S$ is a
    singleton $\{s\}$, then we write $\langle S \rangle = \langle s \rangle$.
\end{definition}

\begin{example}
    Consider an automaton $c : X \to 2 \times X^A$, and denote its transition
    function by $\delta$. For any subset $S \subseteq X$, the subcoalgebra
    $\langle S \rangle$ of $c$ generated by $S$ is obtained by closing $S$
    under the transitions of $c$. Thus the coalgebra $\langle S \rangle$ has
    state space
    \[ \left\{ \delta(s)(u) \mid s \in S, u\in A^* \right\}. \]
    The observations and transitions are inherited from $c$. 
    \label{ex_automatonsubcoalgebra}
\end{example}

The behaviour of an automaton together with an initial state is the language
recognized by that state. Therefore the minimization problem amounts to the
following question: given a language $L$, what is the minimal automaton
recognizing $L$?

\begin{proposition}
    Let $L$ be a language over $A$. The coalgebra with the least number of
    states recognizing $L$ is the subcoalgebra $\langle L \rangle$ of the
    final coalgebra $\pow{A^*}$, with initial state $L \in \langle
    L\rangle$.
\end{proposition}

To generalize this to arbitrary coalgebras, one needs a factorization system
on the underlying category. The abstract monos of the factorization system
play the role of the subcoalgebras. This is worked out in
\cite{abhkms_minimization,roumen}.

\section{Probabilistic systems}\label{sec_probabilisticsystems}

Before representing quantum systems coalgebraically, it is useful to consider
probabilistic systems first, because quantum mechanical behaviour is
probabilistic.

Transitions to a probability distribution over successor states will be
modeled using functors involving the so-called \emph{distribution monad}.
Define a functor $\D : \Sets \to \Sets$ sending a set $X$ to the set of finite
convex combinations, or probability distributions, on $X$:
\[ 
\textstyle
    \D(X) = \{ \varphi : X \to [0,1] \mid \supp(\varphi) \text{ is finite and }
    \sum_{x\in X} \varphi(x) = 1 \}. 
\]
Here $\supp(\varphi) = \{x \in X \mid \varphi(x) \neq 0 \}$ is the
\emph{support} of $\varphi$.
An element $\varphi$ of $\D(X)$ can also be written as a formal sum $r_1 x_1 +
\cdots + r_n x_n$, where $\{x_1,\ldots, x_n\} = \supp(\varphi)$ and $r_i =
\varphi(x_i)$. On a morphism $f : X \to Y$, the functor $\D$ is defined as
\[ \D(f)(r_1 x_1 + \cdots + r_n x_n) = r_1 f(x_1) + \cdots + r_n f(x_n). \]
The functor $\D$ is a monad with unit and multiplication
\[
\begin{array}{rclcrcl}
    \eta : X & \to & \D(X) & \qquad & \mu : \D(\D(X)) & \to & \D(X) \\
    x & \mapsto & 1x & \qquad & \sum_i r_i (\sum_j s_{ij} x_{ij}) & \mapsto &
    \sum_{i,j} r_i s_{ij} x_{ij}
\end{array}
\]

We will show how the distribution monad is used in the representation of
probabilistic systems by an example.

\begin{example}
    \label{ex_markovchain}
    Imagine a particle moving on the following graph.
    \[ \entrymodifiers={++[o][F-]}
    \SelectTips{cm}{}
    
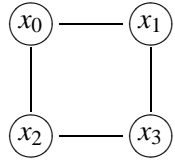
\xymatrix @-1pc { {x_0} \ar@{-}[rr] \ar@{-}[dd] & *\txt{} & {x_1}
    \ar@{-}[dd] \\
    *\txt{} \\
    {x_2} \ar@{-}[rr] & *\txt{} & {x_3} } \]
    Let $X = \{x_0, x_1, x_2, x_3\}$ be the set of vertices. The particle
    starts at one of the vertices of the
    graph.  In each time step, the particle can move to one of the two
    adjacent points. Each of these points is chosen with probability
    $\frac{1}{2}$. This system can be written as a coalgebra for the
    distribution monad:
    \[ \begin{array}{rcl}
        c :  X &\to& \D(X) \\
         x_0 &\mapsto& \tfrac{1}{2}x_1 + \tfrac{1}{2}x_2 \\
         x_1 &\mapsto& \tfrac{1}{2}x_0 + \tfrac{1}{2}x_3 \\
         x_2 &\mapsto& \tfrac{1}{2}x_0 + \tfrac{1}{2}x_3 \\
         x_3 &\mapsto& \tfrac{1}{2}x_1 + \tfrac{1}{2}x_2
    \end{array} \]
    A coalgebra for $\D$ is called a \emph{Markov chain}.

    We consider the trajectory of the particle when it starts in the vertex
    $x_0$. Let $\varphi_n \in \D(X)$ denote the probability distribution over
    the vertices of the graph after $n$ steps. Then the first few values of
    $\varphi_n$ are:
    \[ \begin{array}{rcl}
        \varphi_0 &=& x_0 \\
        \varphi_1 &=& \tfrac{1}{2}x_1 + \tfrac{1}{2}x_2 \\
        \varphi_2 &=& \tfrac{1}{2}x_0 + \tfrac{1}{2}x_3 \\
        \varphi_3 &=& \tfrac{1}{2}x_1 + \tfrac{1}{2}x_2
    \end{array} \]
    We obtain a repetition after three steps, so $\varphi_{2n} = \varphi_2$
    for $n > 0$ and $\varphi_{2n+1} = \varphi_1$ for all $n\in \mathbb{N}$.
\end{example}

It is reasonable to view the sequence $(\varphi_n)_{n \in \mathbb{N}}$ as the
behaviour of the Markov chain. However, if we model the system as a
$\D$-coalgebra, the morphism into the final coalgebra does not give the
desired behaviour. This can be solved by replacing the underlying category
$\Sets$ of the coalgebra by another category. There are at least two possible
replacements known in the literature: \cite{jacobs_generictrace} proposes to
work in the Kleisli category of the monad $\D$, and \cite{sbbr_powersetcoalg}
proposes to work in the category of Eilenberg-Moore algebras for $\D$. The
solutions work for coalgebras for several monads, not only for the
distribution monad. Both approaches are compared in \cite{jss_tracedet}. Here
we will briefly describe how to model probabilistic systems in the category of
Eilenberg-Moore algebras, since this approach will also be used for quantum
systems.

An Eilenberg-Moore algebra for the distribution monad is called a convex set.
In a convex set $X$, we can assign to each convex combination $\sum_{i=1}^n
r_i = 1$ a function $X^n \to X$ denoted $(x_1,\ldots,x_n) \mapsto \sum_{i=1}^n
r_i x_i$.  A homomorphism of convex sets preserves all convex combinations and
is called a convex or affine map. The category of convex sets and maps is
written as $\Conv$, and is also described in \cite{jacobsconvex}.

The next result ensures that several functors on the category $\Conv$ have a
final coalgebra, which enables us to speak about the behaviour of
probabilistic systems.

\begin{lemma}
    \label{lem_generalizedautomata}
    Let $\cat$ be a category with all products. Fix an object $B\in \cat$ and
    a set $A$. The functor $F : \cat \to \cat$ defined by $F(X) = B \times
    X^A$, where $X^A$ denotes a power, has a final coalgebra whose underlying
    state space is the power $B^{A^*}$.
\end{lemma}
\begin{proof}
    The projection $B^{A^*} \to B$ onto the coordinate with index $u\in A^*$
    will be denoted $\pi_u$. The dynamics is a map $\omega : B^{A^*} \to B
    \times (B^{A^*})^A$, whose first component is $\pi_{\varepsilon}$ and
    whose second component in the coordinates $a\in A$, $u\in A^*$ is
    $\pi_{au}$. We will now prove the finality of $\omega$. Given a coalgebra
    $c = \langle f, \langle g_a \rangle_{a\in A} \rangle : X \to B\times X^A$,
    define $\beh{c} : X \to B^{A^*}$ as follows: first extend the family of
    morphisms $(g_a)_{a\in A}$ for $a\in A$ to a family $(g_u)_{u\in A^*}$ by
    defining inductively
    \[ \begin{array}{rcl}
        g_{\varepsilon} &=& \id, \\
        g_{au} &=& g_u \circ g_a.
    \end{array} \]
    Then $\beh{c} = \langle f \circ g_u \rangle_{u\in A^*}$ is a coalgebra
    morphism from $c$ to $\omega$.

    To show that $\beh{c}$ is the unique such coalgebra morphism, let $\varphi
    : X \to B^{A^*}$ be any coalgebra morphism, and denote the component with
    coordinate $u\in A^*$ by $\varphi_u$.  Induction on the word $u\in A^*$
    proves that $\varphi_u = f \circ g_u$.
\end{proof}

\begin{example}
    \label{ex_markovcoalgebra}
    The coalgebra from Example~\ref{ex_markovchain} can also be represented as
    a coalgebra in $\Conv$, in such a way that the morphism into the final
    coalgebra yields the list of probability distributions associated to the
    Markov chain. Fix a set $X$ and take the functor $F(Y) = \D(X) \times Y$
    on the category $\Conv$, where $X$ is the set of vertices of the graph.
    Represent the Markov chain as an $F$-coalgebra $d$ with state space
    $\D(X)$, i.e. a convex map $\D(X) \to \D(X) \times \D(X)$. Since $\D(X)$
    is the free convex set on $X$, it suffices to define $d$ on the set $X$ of
    generators.  Let $d(x) = (1x,c(x)) \in \D(X) \times \D(X)$. From
    Lemma~\ref{lem_generalizedautomata} it follows that the final
    $F$-coalgebra exists and has state space $\D(X)^{\mathbb{N}}$, so we
    obtain a map $\beh{d} : \D(X) \to \D(X)^{\mathbb{N}}$. This behaviour map
    sends the initial state $x_0 \in \D(X)$ to the list of probability
    distributions $(\varphi_n)_{n\in \mathbb{N}}$ obtained in
    Example~\ref{ex_markovchain}.
\end{example}

\section{Coalgebraic model of quantum computation}\label{sec_model}

In this section we will show how to model discrete quantum systems as
coalgebras. We will first introduce two running examples. The first example is
the class of quantum automata. This quantum analogue of deterministic automata
was defined in \cite{moore}. There are two differences between our definition
and \cite{moore}: we generalize the output projections to effects on a Hilbert
space, and we ignore initial states, which is often more convenient for
coalgebras.

\begin{definition}
    A \emph{quantum language} over an alphabet $A$ is a function $A^* \to
    [0,1]$. One can think of a quantum language as a fuzzy or probabilistic
    language.  The function assigns to a word in $A^*$ the probability that it
    is in the language.

    A \emph{quantum automaton} over an alphabet $A$ consists of:
    \begin{itemize}
        \item A complex Hilbert space $H$;
        \item For each letter $a\in A$, a unitary operator $\delta_a : H \to
            H$;
        \item An effect $\effect$ on $H$, i.e. a linear operator $\effect : H
            \to H$ satisfying $0 \leq \effect \leq \id$.
    \end{itemize}
    Define, for each word $u\in A^*$, the extended transition operator
    $\delta_u : H \to H$ inductively by
    \[ \delta_\varepsilon(\psi) = \psi, \]
    \[ \delta_{au}(\psi) = \delta_u(\delta_a(\psi)). \]
    The probability that the word $u$ is accepted by the automaton, starting
    from initial state $\psi$ with norm 1, is
    \[\bk{\delta_u \psi}{\effect}{\delta_u \psi}, \] that is, the
    probability that measurement of $\effect$ gives `yes' in state $\delta_u
    \psi$.

    The quantum language \emph{recognized} by a state $\psi \in H$ is the
    function $A^* \to [0,1]$ that sends a word $u\in A^*$ to the probability
    that $u$ is accepted by $\psi$.
\end{definition}

The quantum walks form another class of examples of systems, discussed
extensively in \cite{elias}. See also \cite{jacobswalks} for coalgebraic
versions. We will take the example from the latter reference.

\begin{example}
    \label{ex_qwalk}
    Consider a particle walking on the line $\mathbb{Z}$ of integers. In
    addition to the position of the particle on $\mathbb{Z}$, we take its spin
    into account. 
    Then the Hilbert space modeling the composite system
    of the particle's spin and position is $\mathbb{C}^2 \otimes
    \ell^2(\mathbb{Z})$. Write the basis vectors in Dirac notation as $\ku{k}$
    and $\kd{k}$.  We stipulate that the particle starts in state $\ku{0}$.
    The dynamics of the walk is given by the unitary operator
    \[ \begin{array}{rrcl}
        U: & \ku{k} &\mapsto& \ost \ku{k-1} + \ost \kd{k+1} \\
           & \kd{k} &\mapsto& \ost \ku{k-1} - \ost \kd{k+1}
    \end{array} \]
    Consider a situation in which it is possible to measure the position of
    the particle, but not the spin. Then the admissible observables are
    $\ku{k}\bu{k} + \kd{k}\bd{k}$ for $k\in \mathbb{Z}$.

    If the particle walks during $n$ time steps, then its position can be
    described using a probability distribution over $\mathbb{Z}$. The
    probability that we encounter the particle on position $k$ is 
    \[ \bk{U^n (\uparrow \! 0)} {\big(\ku{k}\bu{k} + \kd{k}\bd{k}\big)} {U^n
    (\uparrow \! 0)}. \]
    Denote the probability distribution after $n$ steps by $\varphi_n \in
    \D(\mathbb{Z})$. The unit distribution $k$ is written using Dirac notation
    as $\ket{k}$. The first few probability distributions are:
    \[ \begin{array}{rcccccccc}
        \varphi_0 &=& & & & \ket{0} \\
        \varphi_1 &=& & & \frac{1}{2}\ket{-1} &+& \frac{1}{2}\ket{1} \\
        \varphi_2 &=& & \frac{1}{4}\ket{-2} &+& \frac{1}{2}\ket{0} &+&
        \frac{1}{4}\ket{2} \\
        \varphi_3 &=& \frac{1}{8}\ket{-3} &+& \frac{5}{8}\ket{-1} &+&
        \frac{1}{8}\ket{1} &+& \frac{1}{8}\ket{3}
    \end{array} \]
\end{example}

We would like to model quantum systems in such a way that the system is a
coalgebra for a certain endofunctor, and the morphism into the final coalgebra
gives the quantum language or probability distribution determined by the
system. To achieve this, we will work with coalgebras in the category $\Conv$
of convex algebras. There are two reasons for this. Firstly, we can model
computations for both systems in pure states and systems in mixed states.
Secondly, the category $\Conv$ incorporates both quantum probability via
density matrices and the classical probability that is needed for the output.

Let $\Hilbisometry$ be the category of Hilbert spaces with isometric
embeddings as morphisms. Define a functor $ \DM : \Hilbisometry \to \Conv $ on
objects by letting $\DM(H)$ be the set of density matrices on the Hilbert
space $H$, and on morphisms by $\DM(f: H \to K)(\rho) = f \rho f^\dag $.

Let $H$ be a Hilbert space underlying a quantum system, and let $\unops$ be a
set of unitary operators that can be applied to the system. Represent the
possible measurements on the system by a subset $\meass$ of the set of effects
$\Ef(H)$ on $H$. We do
not use the entire set $\Ef(H)$ since usually not all effects are possible or
interesting. It is often the case that the sum of the effects in $\meass$ is
$\id$. In this case, the subset $\meass$ is called a \emph{test}, see
\cite{dvurecenskij} for more information.  Consider the functor 
\begin{equation}
    F(X) = [0,1]^{\meass} \times X^{\unops}
    \label{eq_schrofunctor}
\end{equation}
on the category $\Conv$. Form the $F$-coalgebra
\begin{equation}
    \label{eq_schrocoalg} 
    \begin{array}{rcl}
    f:  \DM(H) &\to& [0,1]^{\meass} \times \DM(H)^{\unops} \\
     \rho &\mapsto& \left( \left(\tr(\rho \effect) \right)_{\effect \in
    \meass}, \left( \DM(U)(\rho) \right)_{U\in \unops} \right).
\end{array}
\end{equation}
The part $\tr(\rho \effect)$ represents the observations on the coalgebra, since
this is the probability that measurement of the effect $\effect$ succeeds when
the system is in mixed state $\rho$. The part $\DM(U)(\rho)$ is the evolution
of the system according to the density matrix formalism.

We will now show that \eqref{eq_schrocoalg} has the desired behaviour. First
we have to check that the behaviour exists, i.e. that $F$ has a final
coalgebra. 

The category $\Conv$ inherits all products from $\Sets$, so
Lemma~\ref{lem_generalizedautomata} applies to the functor defined in
\eqref{eq_schrofunctor}. The functor $F$ has a final coalgebra with the convex
set $([0,1]^{\meass})^{\unops^*}$ as state space. The dynamics is a map
$([0,1]^{\meass})^{\unops^*} \to [0,1]^{\meass} \times (
([0,1]^{\meass})^{\unops^*})^{\unops}$. The first component of the dynamics is
the projection onto the component with the empty word $\varepsilon$ as index,
and for the second component, the map with index $a \in \unops$ and $u\in
\unops^*$ is the projection onto component $au$.

The coalgebra \eqref{eq_schrocoalg} gives a behaviour map $\beh{f} : \DM(H)
\to ([0,1]^{\meass})^{\unops^*}$. This map can alternatively be seen as a
indexed family of maps $\DM(H) \to [0,1]$ for each $\effect \in \meass$ and
$u\in \unops^*$. The map with index $\effect$ and $u = u_1 \ldots u_n \in
\unops^*$ sends the density matrix $\rho$ to $\tr(u_n \ldots u_1 \rho u_1^\dag
\ldots u_n^\dag \effect)$. Physically, if we view effects as yes-no questions
about a system, this is the probability that measurement of the effect
$\effect$ yields outcome `yes', if the system starts in mixed state $\rho$ and
evolves according to the unitary operators $u_1$, \dots, $u_n$. Therefore the
final coalgebra semantics corresponds exactly to the physical behaviour.

The construction of the coalgebra \eqref{eq_schrocoalg} involves three
parameters: the underlying state space $H$, the set of unitary operators
$\unops$, and the set of possible measurements $\meass$. By choosing these
parameters appropriately we can fit the above examples in this framework.

\begin{example}
    \label{ex_quantumsystems}
    Let $(H,(\delta_a)_{a \in A}, \effect)$ be a quantum automaton. The set of
    unitaries $\unops$ is $\{\delta_a \mid a \in A\}$, and the set of possible
    measurements is the singleton set $\meass = \{\effect\}$. The coalgebra
    \eqref{eq_schrocoalg} becomes
    \[ \begin{array}{rcl}
        \DM(H) &\to& [0,1] \times \DM(H)^A \\
        \rho &\mapsto& \left( \tr(\rho \effect), \left( \delta_a \rho
        \delta_a^\dag \right)_{a\in A} \right)
    \end{array} \]
    If $\psi \in H$ is a state, then the behaviour map $\DM(H) \to
    [0,1]^{A^*}$ maps the associated density matrix $\ket{\psi}\bra{\psi}$ to
    the quantum language recognized by $\psi$.

    The quantum walk from Example~\ref{ex_qwalk} gives a coalgebra
    \[ \begin{array}{rcl}
        \DM(\mathbb{C}^2 \otimes \ell^2(\mathbb{Z})) &\to& [0,1]^{\mathbb{Z}}
        \times \DM(\mathbb{C}^2 \otimes \ell^2(\mathbb{Z})) \\
        \rho &\mapsto& \left( \left( \tr(\rho \ku{k}\bu{k} + \rho
        \kd{k}\bd{k}) \right)_{k\in \mathbb{Z}} , U \rho U^\dag \right).
    \end{array} \]
    The codomain can be restricted to $\D_{\omega}(\mathbb{Z}) \times
    \DM(\mathbb{C}^2 \otimes \ell^2(\mathbb{Z}))$, because the effects in
    $\meass$ sum to $\id$. Here $\D_\omega$ is the infinite distribution
    monad, defined on objects by
    \[ \textstyle
    \D_\omega(X) = \{ \varphi : X \to [0,1]
    \mid \supp(\varphi) \text{ is at most countable and } \sum_{x\in X}
    \varphi(x) = 1 \}. \]
    On morphisms, $\D_\omega$ acts the same as the ordinary distribution
    monad $\D$.
\end{example}

These examples show that the representation \eqref{eq_schrocoalg} captures
many quantum systems in a uniform way.

\section{Minimization}\label{sec_minimization}

In Section~\ref{sec_coalgebras} we presented a procedure to find minimal
automata for classical languages. The same method can be used to find minimal
realizations for quantum behaviour.

Consider coalgebras for the functor $F(X) = [0,1]^E \times X^S$, which was
defined in \eqref{eq_schrofunctor}. The final coalgebra for $F$ has underlying
state space $([0,1]^E)^{S^*}$. Then the minimal realization of a behaviour $L
\in ([0,1]^E)^{S^*}$ is the smallest subcoalgebra of $([0,1]^E)^{S^*}$
containing $L$.  An interesting feature of this approach is that the minimal
coalgebra need not be a quantum coalgebra, even if the behaviour arises from a
quantum mechanical system. Since the subcoalgebra of the final coalgebra lies
in the category $\Conv$, the minimal realization will certainly be a
probabilistic system, but its state space is not necessarily of the form
$\DM(H)$. Thus minimization gives a procedure for turning quantum systems into
simpler systems in $\Conv$ which nevertheless have the same behaviour.

\begin{example}
    \label{ex_miminization}
    We consider a quantum version of Example~\ref{ex_markovchain}. The square
    graph gives rise to the state space $\mathbb{C}^4$ with basis $\ket{0}$,
    $\ket{1}$, $\ket{2}$, $\ket{3}$, and the walk of a quantum particle on the
    square graph is represented by the unitary matrix
    \[ U = \frac{1}{\sqrt{2}} \left(
    \begin{array}{cccc}
        0 & 1 & 1 & 0 \\
        1 & 0 & 0 & 1 \\
        1 & 0 & 0 & -1 \\
        0 & 1 & -1 & 0
    \end{array}
    \right). \]
    For each vertex $k$, we can measure the probability that the particle is
    in state $k$ using the projection $\ket{k}\bra{k}$. This leads to the
    coalgebra
    \[ \begin{array}{rcl}
        f:  \DM(\mathbb{C}^4) &\to& \D(4) \times \DM(\mathbb{C}^4)  \\
         \rho &\mapsto& \left( \left(\tr(\rho \ket{k}\bra{k}) \right)_{
        k=0,1,2,3}, U \rho U^\dag \right).
    \end{array} \]
    Assume that $\ket{0}\bra{0} \in \DM(\mathbb{C}^4)$ is the initial state.
    The resulting output stream is 
    \[ \sigma = \left( \ket{0}, \frac{1}{2}\ket{1} + \frac{1}{2}\ket{2}
    \right)^{\omega}. \]
    The minimal coalgebra with this behaviour is the smallest subcoalgebra of
    $\D(4)^{\mathbb{N}}$ containing $\sigma$. We shall compute this
    subcoalgebra by determining the states of $\D(4)^{\mathbb{N}}$ that can be
    reached from $\sigma$ with the transition structure, and then taking the
    convex set generated by the reachable states.

    The transition structure of the final coalgebra is given by the function
    $\tail: \D(4)^{\mathbb{N}} \to \D(4)^{\mathbb{N}}$ defined by
    $\tail(x_0,x_1,x_2,\ldots) = (x_1,x_2,\ldots)$. Applying the $\tail$
    function repeatedly to $\sigma$ gives the streams $\sigma$ and $\sigma' =
    \left( \frac{1}{2}\ket{1} + \frac{1}{2}\ket{2}, \ket{0} \right)^{\omega}$.
    The convex algebra generated by $\sigma$ and $\sigma'$ inside
    $\D(4)^{\mathbb{N}}$ is
    \[ 
    X = \left\{ \left( \begin{array}{c}
        p\ket{0} + \tfrac{1}{2}(1-p)\ket{1} + \tfrac{1}{2}(1-p)\ket{2}, \\ 
        (1-p)\ket{0} + \tfrac{1}{2}p\ket{1} + \tfrac{1}{2}p\ket{2}
        \end{array}
        \right)^{\omega} \ \bigg| \  p\in [0,1] \right\}.
    \]
    This is the minimal coalgebra with behaviour $\sigma$ and hence the
    minimization of $f$. The coalgebra structure is obtained as restriction of
    the final coalgebra $\D(4)^{\mathbb{N}}$.  By projecting onto the first
    coordinate twice we obtain an affine isomorphism between $X$ and $[0,1]$.
    Therefore a more elementary display of the minimization is
    \[ \begin{array}{rcl}
        [0,1] &\to& \D(4) \times [0,1] \\
        p &\mapsto& \left( p\ket{0} + \tfrac{1}{2}(1-p)\ket{1} +
        \tfrac{1}{2}(1-p)\ket{2}, 1-p \right)
    \end{array} \]
    Observe that the state space $[0,1]$ of the minimization is more
    manageable than the original state space $\DM(\mathbb{C}^4)$. The
    minimization is not a quantum system anymore, since $[0,1]$ is not
    isomorphic to a convex set of density matrices. It can be seen as a
    probabilistic system with two states, since $[0,1] \cong \D(2)$. Thus the
    behaviour of this quantum mechanical system can be reproduced with a
    simpler probabilistic system. Note that we only consider the classically
    observable probabilities as part of the behaviour, not the quantum
    mechanical amplitudes. Therefore the minimization only contains
    information about the classically probabilistic behaviour.
\end{example}

\section{Conclusion}
\label{sec_conclusion}

Quantum systems can be represented by using density matrices as states, with
unitary operators acting on them as dynamics. We have exploited this fact to
model quantum systems as coalgebras in the category $\Conv$ of convex sets,
with a set of density matrices as state space. A consequence of this
representation is that minimization of coalgebras gives a method to transform
a quantum system into a probabilistic system that has the same probabilistic
behaviour, but a simpler structure. 

There are several possibilities for future research. We have only studied
minimization of quantum systems empirically through examples. It would be nice
to have more general results about the structure of minimal realizations.
Moreover it is unclear if quantum systems can also be modeled in such a way
that minimization gives a quantum coalgebra again.

\paragraph{Acknowledgements.} 
This work is based on my master's thesis \cite{roumen}.
I would like to thank my thesis supervisor Bart Jacobs for many helpful
discussions and comments.

\bibliographystyle{eptcs}
\bibliography{quantumcoalg}

\begin{thebibliography}{10}
\providecommand{\bibitemdeclare}[2]{}
\providecommand{\surnamestart}{}
\providecommand{\surnameend}{}
\providecommand{\urlprefix}{Available at }
\providecommand{\url}[1]{\texttt{#1}}
\providecommand{\href}[2]{\texttt{#2}}
\providecommand{\urlalt}[2]{\href{#1}{#2}}
\providecommand{\doi}[1]{doi:\urlalt{http://dx.doi.org/#1}{#1}}
\providecommand{\bibinfo}[2]{#2}

\bibitemdeclare{inproceedings}{abramsky_chucoalg}
\bibitem{abramsky_chucoalg}
\bibinfo{author}{S.~\surnamestart Abramsky\surnameend} (\bibinfo{year}{2010}):
  \emph{\bibinfo{title}{Coalgebras, \protect{C}hu spaces, and representations
  of physical systems}}.
\newblock In: {\sl \bibinfo{booktitle}{Proceedings of the 25th Annual IEEE
  Symposium on Logic in Computer Science}}, pp. \bibinfo{pages}{411--420},
  \doi{10.1109/LICS.2010.35}.

\bibitemdeclare{inproceedings}{ac_protocols}
\bibitem{ac_protocols}
\bibinfo{author}{S.~\surnamestart Abramsky\surnameend} \&
  \bibinfo{author}{B.~\surnamestart Coecke\surnameend} (\bibinfo{year}{2004}):
  \emph{\bibinfo{title}{A categorical semantics of quantum protocols}}.
\newblock In: {\sl \bibinfo{booktitle}{Proceedings of the 19th Annual IEEE
  Symposium on Logic in Computer Science}}, pp. \bibinfo{pages}{415--425},
  \doi{10.1109/LICS.2004.1319636}.

\bibitemdeclare{incollection}{abhkms_minimization}
\bibitem{abhkms_minimization}
\bibinfo{author}{J.~\surnamestart Ad\'amek\surnameend},
  \bibinfo{author}{F.~\surnamestart Bonchi\surnameend},
  \bibinfo{author}{M.~\surnamestart H\"ulsbusch\surnameend},
  \bibinfo{author}{B.~\surnamestart K\"onig\surnameend},
  \bibinfo{author}{S.~\surnamestart Milius\surnameend} \&
  \bibinfo{author}{A.~\surnamestart Silva\surnameend} (\bibinfo{year}{2012}):
  \emph{\bibinfo{title}{A coalgebraic perspective on minimization and
  determinization}}.
\newblock In \bibinfo{editor}{Lars \surnamestart Birkedal\surnameend}, editor:
  {\sl \bibinfo{booktitle}{Foundations of Software Science and Computational
  Structures}}, {\sl \bibinfo{series}{Lecture Notes in Computer Science}}
  \bibinfo{volume}{7213}, \bibinfo{publisher}{Springer Berlin Heidelberg}, pp.
  \bibinfo{pages}{58--73}, \doi{10.1007/978-3-642-28729-9\_4}.

\bibitemdeclare{article}{bi_topos}
\bibitem{bi_topos}
\bibinfo{author}{J.~\surnamestart Butterfield\surnameend} \&
  \bibinfo{author}{C.~\surnamestart Isham\surnameend} (\bibinfo{year}{1998}):
  \emph{\bibinfo{title}{A topos perspective on the
  \protect{K}ochen--\protect{S}pecker theorem: \protect{I}. \protect{Q}uantum
  states as generalized valuations}}.
\newblock {\sl \bibinfo{journal}{Int. Jour. Theor. Physics}}
  \bibinfo{volume}{37}(\bibinfo{number}{11}), pp. \bibinfo{pages}{2669--2733},
  \doi{10.1023/A:1026680806775}.

\bibitemdeclare{book}{dvurecenskij}
\bibitem{dvurecenskij}
\bibinfo{author}{A.~\surnamestart Dvure\v{c}enskij\surnameend} \&
  \bibinfo{author}{S.~\surnamestart Pulmannov\'a\surnameend}
  (\bibinfo{year}{2000}): \emph{\bibinfo{title}{New Trends in Quantum
  Structures}}.
\newblock \bibinfo{publisher}{Kluwer Academic Publishers},
  \doi{10.1007/978-94-017-2422-7}.

\bibitemdeclare{article}{goguen}
\bibitem{goguen}
\bibinfo{author}{J.~\surnamestart Goguen\surnameend} (\bibinfo{year}{1972}):
  \emph{\bibinfo{title}{Minimal realizations of machines in closed
  categories}}.
\newblock {\sl \bibinfo{journal}{Bull. Amer. Math. Soc.}}
  \bibinfo{volume}{78}(\bibinfo{number}{5}), pp. \bibinfo{pages}{777--783},
  \doi{10.1090/S0002-9904-1972-13032-5}.

\bibitemdeclare{article}{jacobs_generictrace}
\bibitem{jacobs_generictrace}
\bibinfo{author}{I.~\surnamestart Hasuo\surnameend},
  \bibinfo{author}{B.~\surnamestart Jacobs\surnameend} \&
  \bibinfo{author}{A.~\surnamestart Sokolova\surnameend}
  (\bibinfo{year}{2007}): \emph{\bibinfo{title}{Generic trace semantics via
  coinduction}}.
\newblock {\sl \bibinfo{journal}{Logical Methods in Computer Science}}
  \bibinfo{volume}{3}(\bibinfo{number}{4:11}), pp. \bibinfo{pages}{1--36}.

\bibitemdeclare{incollection}{jacobsconvex}
\bibitem{jacobsconvex}
\bibinfo{author}{B.~\surnamestart Jacobs\surnameend} (\bibinfo{year}{2010}):
  \emph{\bibinfo{title}{Convexity, duality and effects}}.
\newblock In \bibinfo{editor}{Cristian \surnamestart Calude\surnameend} \&
  \bibinfo{editor}{Vladimiro \surnamestart Sassone\surnameend}, editors: {\sl
  \bibinfo{booktitle}{Theoretical Computer Science}}, \bibinfo{series}{IFIP
  Advances in Information and Communication Technology},
  \bibinfo{publisher}{Springer Boston}, pp. \bibinfo{pages}{1--19},
  \doi{10.1007/978-3-642-15240-5\_1}.

\bibitemdeclare{inproceedings}{jacobswalks}
\bibitem{jacobswalks}
\bibinfo{author}{B.~\surnamestart Jacobs\surnameend} (\bibinfo{year}{2011}):
  \emph{\bibinfo{title}{Coalgebraic walks, in quantum and \protect{T}uring
  computing}}.
\newblock In: {\sl \bibinfo{booktitle}{Lecture Notes in Computer Science}},
  {\sl \bibinfo{series}{Foundations of Software Science and Computational
  Structures}} \bibinfo{volume}{6604}, pp. \bibinfo{pages}{12--26},
  \doi{10.1007/978-3-642-19805-2\_2}.

\bibitemdeclare{inproceedings}{jss_tracedet}
\bibitem{jss_tracedet}
\bibinfo{author}{B.~\surnamestart Jacobs\surnameend},
  \bibinfo{author}{A.~\surnamestart Silva\surnameend} \&
  \bibinfo{author}{A.~\surnamestart Sokolova\surnameend}
  (\bibinfo{year}{2012}): \emph{\bibinfo{title}{Trace semantics via
  determinization}}.
\newblock In: {\sl \bibinfo{booktitle}{Lecture Notes in Computer Science, to
  appear}}, \bibinfo{series}{Proceedings of CMCS 2012},
  \doi{10.1007/978-3-642-32784-1\_7}.

\bibitemdeclare{article}{moore}
\bibitem{moore}
\bibinfo{author}{C.~\surnamestart Moore\surnameend} \&
  \bibinfo{author}{J.~\surnamestart Crutchfield\surnameend}
  (\bibinfo{year}{2000}): \emph{\bibinfo{title}{Quantum automata and quantum
  grammars}}.
\newblock {\sl \bibinfo{journal}{Theoretical Computer Science}}
  \bibinfo{volume}{237}(\bibinfo{number}{1}), pp. \bibinfo{pages}{275--306},
  \doi{10.1016/S0304-3975(98)00191-1}.

\bibitemdeclare{article}{nerode}
\bibitem{nerode}
\bibinfo{author}{A.~\surnamestart Nerode\surnameend} (\bibinfo{year}{1958}):
  \emph{\bibinfo{title}{Linear automaton transformations}}.
\newblock {\sl \bibinfo{journal}{Proc. Amer. Math. Soc.}}
  \bibinfo{volume}{9}(\bibinfo{number}{4}), pp. \bibinfo{pages}{541--544},
  \doi{10.1090/S0002-9939-1958-0135681-9}.

\bibitemdeclare{mastersthesis}{roumen}
\bibitem{roumen}
\bibinfo{author}{F.~\surnamestart Roumen\surnameend} (\bibinfo{year}{2012}):
  \emph{\bibinfo{title}{Coalgebraic semantics for quantum computation}}.
\newblock Master's thesis, \bibinfo{school}{Radboud University Nijmegen}.

\bibitemdeclare{inproceedings}{ruttenautomatacoinduction}
\bibitem{ruttenautomatacoinduction}
\bibinfo{author}{J.~\surnamestart Rutten\surnameend} (\bibinfo{year}{1998}):
  \emph{\bibinfo{title}{Automata and coinduction (an exercise in coalgebra)}}.
\newblock In \bibinfo{editor}{D.~\surnamestart Sangiorigi\surnameend} \&
  \bibinfo{editor}{R.~\surnamestart de~Simone\surnameend}, editors: {\sl
  \bibinfo{booktitle}{Proceedings of CONCUR '98}}, {\sl \bibinfo{series}{LNCS}}
  \bibinfo{volume}{1466}, \bibinfo{publisher}{Springer}, pp.
  \bibinfo{pages}{194--218}, \doi{10.1007/BFb0055624}.

\bibitemdeclare{article}{rutten_universalcoalg}
\bibitem{rutten_universalcoalg}
\bibinfo{author}{J.~\surnamestart Rutten\surnameend} (\bibinfo{year}{2000}):
  \emph{\bibinfo{title}{Universal coalgebra: a theory of systems}}.
\newblock {\sl \bibinfo{journal}{Theoretical Computer Science}}
  \bibinfo{volume}{249}(\bibinfo{number}{1}), pp. \bibinfo{pages}{3--80},
  \doi{10.1016/S0304-3975(00)00056-6}.

\bibitemdeclare{inproceedings}{sbbr_powersetcoalg}
\bibitem{sbbr_powersetcoalg}
\bibinfo{author}{A.~\surnamestart Silva\surnameend},
  \bibinfo{author}{F.~\surnamestart Bonchi\surnameend},
  \bibinfo{author}{M.~\surnamestart Bonsangue\surnameend} \&
  \bibinfo{author}{J.~\surnamestart Rutten\surnameend} (\bibinfo{year}{2010}):
  \emph{\bibinfo{title}{Generalizing the powerset construction,
  coalgebraically}}.
\newblock In: {\sl \bibinfo{booktitle}{Proceedings of the IARCS Annual
  Conference on Foundations of Software Technology and Theoretical Computer
  Science}}, {\sl \bibinfo{series}{Leibniz International Proceedings in
  Informatics}}~\bibinfo{volume}{8}, pp. \bibinfo{pages}{272--283},
  \doi{10.4230/LIPIcs.FSTTCS.2010.272}.

\bibitemdeclare{book}{elias}
\bibitem{elias}
\bibinfo{author}{S.~El\'ias \surnamestart Venegas-Andraca\surnameend}
  (\bibinfo{year}{2008}): \emph{\bibinfo{title}{Quantum Walks for Computer
  Scientists}}.
\newblock \bibinfo{publisher}{Morgan \& Claypool}.

\end{thebibliography}

\end{document}